\documentclass[ 10 pt, conference]{ieeeconf}  

\IEEEoverridecommandlockouts                              

\usepackage{comment}

\usepackage{amsmath,graphicx,amsfonts,amssymb,epsfig,subfig,mathrsfs,mathtools}
\usepackage{cuted}
\usepackage{arydshln}
\usepackage{blkarray}

\usepackage{color}
\usepackage{multirow}
\usepackage{multicol}
\usepackage{lipsum}
\usepackage{rotating}
\usepackage{graphicx}%
\usepackage{algorithm}
\usepackage{algpseudocode}
\usepackage{cite}
\usepackage{tikz}
\usepackage{cancel}
\usepackage{textcomp}

\DeclareMathOperator*{\argmin}{arg\,min}

\newtheorem{assumption}{Assumption}
\newtheorem{theorem}{Theorem}

\newtheorem{proposition}{Proposition}

\allowdisplaybreaks

\newlength\myindent
\def\BibTeX{{\rm B\kern-.05em{\sc i\kern-.025em b}\kern-.08em
    T\kern-.1667em\lower.7ex\hbox{E}\kern-.125emX}}

\title{\LARGE \bf
Ergodic Control Strategy for Multi-Agent Environment Exploration
}

\author{Rabiul Hasan Kabir$^{1}$, Kooktae Lee$^{1}$, and Geronimo Macias$^{1}$
\thanks{$^{1}$R. H. Kabir, K. Lee, and G. Macias are with the Department of Mechanical Engineering, New Mexico Institute of Mining and Technology, Socorro, NM 87801, USA,
        {\tt\scriptsize rabiul.kabir@student.nmt.edu, kooktae.lee@nmt.edu, geronimo.macias@student.nmt.edu}}%
}

\begin{document}

\maketitle
\thispagestyle{empty}
\pagestyle{empty}

\begin{abstract}
In this study, an ergodic environment exploration problem is introduced for a centralized multi-agent system. Given the reference distribution represented by the Mixture of Gaussian (MoG), the ergodicity is achieved when the time-averaged robot distribution is identical to the given reference distribution. The major challenge associated with this problem is to determine a proper timing for a team of agents (robots) to visit each Gaussian component in the reference MoG for ergodicity. The ergodic function is defined as a measure of ergodicity and the condition for convergence is derived based on timing analysis. The proposed control strategy provides relatively reasonable performance to achieve the ergodicity. We provide the formal algorithm for centralized multi-agent control to achieve the ergodicity and simulation results are presented for the validation of the proposed algorithm. 

\end{abstract}

\section{INTRODUCTION}

In recent years, the concept of ergodic environment exploration for autonomous robots has gained a lot of interests due to the efficiency of the exploration scheme as well as its applicability to various sectors such as search and rescue, disaster response, surveillance and reconnaissance, wildlife and weather monitoring, space exploration, etc. The efficiency in this context indicates that the agents are capable of covering a spatial domain with some priority or degrees of importance associated with the given environment. In this case, the agents need to survey an environment such that the time-averaged robot distribution is identical to the given reference distribution.

The first approach to achieve the ergodic exploration for an autonomous multi-agent system is introduced in \cite{mathew2011metrics}. This study provides a new measure for ergodicity based on Fourier Basis Function indicating the difference between the time-averaged robot distribution and reference spatial distribution. This ergodic metric is then utilized to develop a feedback control law for first and second-order robot dynamics with centralized multi-agent systems.  
In \cite{silverman2013optimal}, a strategy to obtain an optimal trajectory for a team of autonomous robots is proposed for data acquisition task. This study focuses on designing an automated trajectory using an optimal control that makes the agents spend time in a region where the duration of time is proportional to the probability of getting informative data. This study also employs the Fourier Basis Function-based ergodic metric to determine the ergodicity of the robot distribution.
A similar work is done by \cite{miller2013trajectory} where the objective is to develop an algorithm to generate trajectories for efficient explorations based on a probabilistic information density map of the given region. Another contribution of this study is the consideration of general non-linear robot dynamics.
In \cite{ayvali2017ergodic}, an ergodic exploration scheme is presented for multiple agents deployed to survey a region with obstacles and restrictive areas, where coordination among multiple agents with different sensory capability is utilized to demonstrate ergodic exploration of the given domain. A finite receding horizon optimal control-based algorithm, termed Ergodic Environmental Exploration (E3) is proposed in \cite{o2015optimal}, to survey an unknown environment consisting of regions with varying degrees of importance. This algorithm determines a required minimum control effort and minimum difference between distributions for time-averaged system trajectory and the information gain.
In \cite{prabhakar2015symplectic}, an iterative optimal control algorithm for general nonlinear dynamics is studied. The authors have demonstrated two separate approaches for discrete-time iterative optimization  -- first order discretization and symplectic integration. It is presented that the control and state trajectories are significantly influenced by the discretization choice for a system. 
A receding horizon control approach to achieve ergodic coverage is proposed in \cite{mavrommati2017real}. This study presents that the algorithm improves the ergodicity between an information density distribution of the spatial domain and time-averaged trajectory of agents. This algorithm enables the agents to explore  the domain independently and to share coverage information with other agents across a communication network. 
A trajectory optimization approach is presented in \cite{de2016ergodic} for ergodic area explorations with the consideration of stochastic nonlinear sensor dynamics. The results suggest that the proposed algorithm can generate trajectories with greater and more predictable ergodicity. 
A decentralized multi-agent ergodic control algorithm with nonlinear dynamics is developed in \cite{abraham2018decentralized} . This algorithm requires the agents to share only a coefficient associated with the action of one agent with others to realize decentralized exploration. 

It is noteworthy that all of the aforementioned research works on the ergodic environment exploration are developed based on the Fourier Basis Function introduced in \cite{mathew2011metrics}. The downside of this metric is that for practical implementation, it inevitably involves an approximation due to the inclusion of infinite summation terms. 
A similar concept related to obtain ergodicity is discussed in \cite{milutinovi2006modeling,hamann2008framework,qi2014multi,ivic2016ergodicity,eren2017velocity}, where ergodicity is achieved based on the global behaviors of multi-agent system from the macrostate of the partial differential equation. However, the appropriate behavior can only be attained if there are extremely large numbers of agents deployed in the domain. Another ergodic exploration scheme inspired by the optimal transport theory is presented in \cite{kabir2020ACC}, which also involves an approximation due to the sample representation of the reference spatial distribution. 

In our previous work \cite{kabir2020ergodicity}, an ergodic environment exploration plan is proposed based on the timing analysis. This plan is, however, only applicable to a single-agent system. In this study, a new strategy to realize a multi-agent ergodic exploration is developed. A Mixture of Gaussian is considered as a reference spatial distribution and the agents are assumed to generate a mass in the form of a skinny Gaussian. The main problem is to find proper timing for the team of agents to survey and exit a Gaussian component. The ergodic function, a measure of ergodicity, is defined and a condition that guarantees the convergence of the ergodic function is derived, which is one of the major contributions of this paper. Further, a control law for the agent position update is provided and the formal algorithm for realizing multi-agent ergodic exploration is presented. To verify the validity of the proposed method, simulations are carried out and simulation results are provided. 

\section{PROBLEM DESCRIPTION}
\textit{Notation:} A set of real and natural numbers are denoted by $\mathbb{R}$ and $\mathbb{N}$, respectively. Further, $\mathbb{N}_0 = \mathbb{N}\cup\{0\}$. The symbols $\Vert \cdot \Vert$ and $^{T}$, respectively, denote the Euclidean norm and the transpose operator. The variable $k\in\mathbb{N}_0$ is used to denote a discrete time.

This section introduces the ergodic environment exploration problem for multi-agent systems. 
The spatial distribution, $\rho^*$, is given as the reference distribution, which is assumed to be in the form of an MoG as follows:
\begin{assumption}
The given spatial distribution $\rho^*$ is expressed as an MoG in the following form:
\begin{align}
    \rho^* = \sum_{i=1}^{m}\alpha_i\mathcal{N}(\mu_i, \Sigma_i),
\end{align}
where $\alpha_i$ is a weight such that $0< \alpha_i < 1$, $\forall i$ with $\sum_{i=1}^{m}\alpha_i=1$, $\mathcal{N}(\mu_i, \Sigma_i)$ is a Gaussian distribution with mean $\mu_i$ and covariance $\Sigma_i$ and $m$ is the total number of Gaussian distributions in the given MoG.
\end{assumption}
Throughout this paper, it is also assumed that $\rho^*$ is stationary (i.e., it does not change over time).

Given $N\in\mathbb{N}$ numbers of agents to explore the domain, each agent is assumed to generate a unit mass concentrated on the current location with a skinny Gaussian distribution as illustrated in Fig. \ref{fig: problem_schematic}.
\begin{figure}[!h]
    \centering
    \includegraphics[scale=0.5]{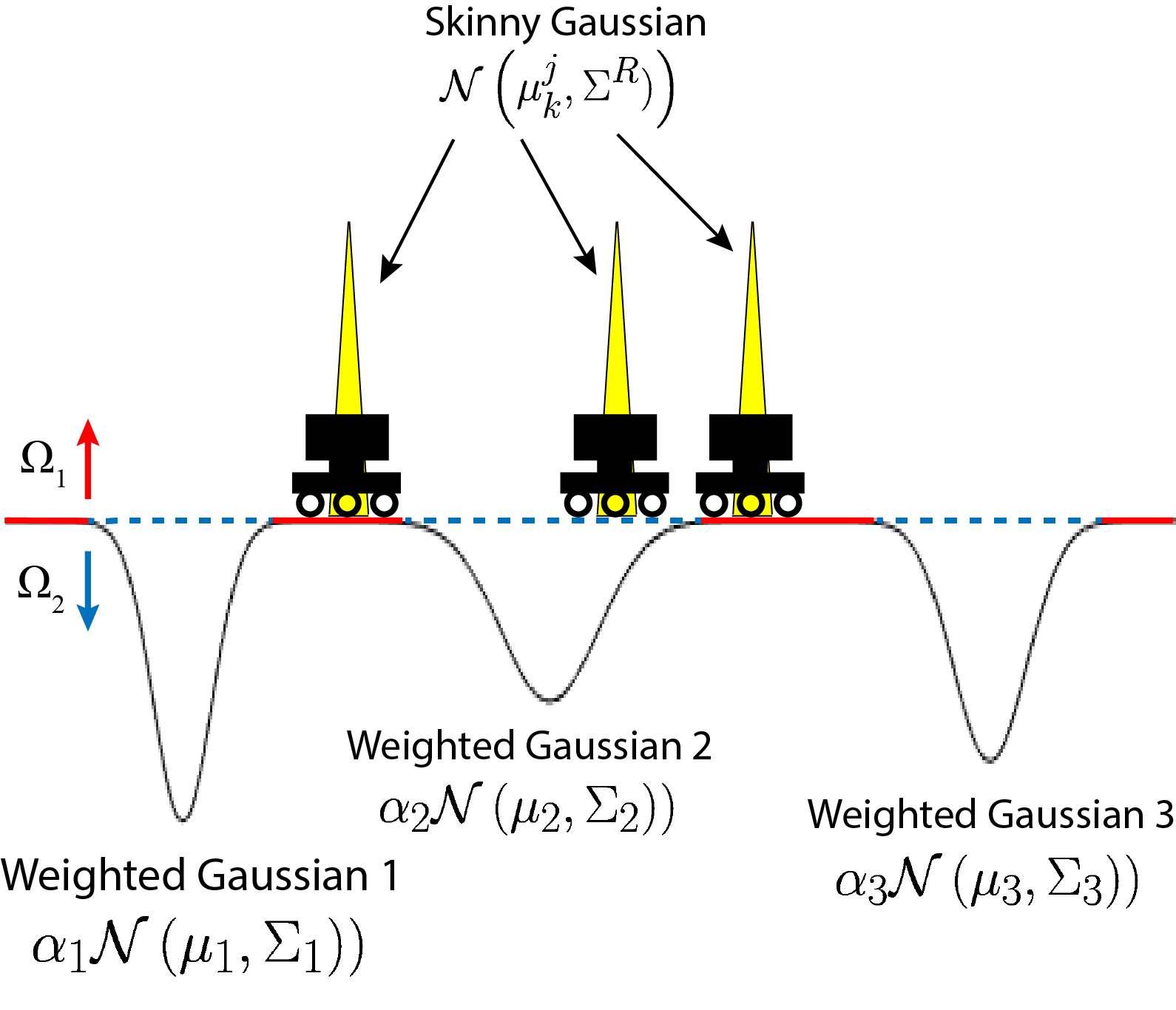}
    \caption{Schematic of the mass generation by the multi-agent system: 1D case}
    \label{fig: problem_schematic}
\end{figure}

In the continuous time, the agent $j$ located at $\mu^j\in\mathbb{R}^2 $ keeps generating a unit mass with a skinny Gaussian distribution, which is generalized as follows.
\begin{assumption}
Suppose that the position of the agent $j$ at any time is given as $\mu^j$. Then, the mass generated by the agent is represented by a skinny Gaussian $f^j:=\mathcal{N}(\mu^j,\Sigma^R)$, where $\Sigma^R$ is stationary, given such that its distribution is narrow and identical for all agents.
\end{assumption}

Mathematically, $f^j$ for the two dimensional case has the following structure:
\begin{equation}\label{eqn: f_{k}}
    \begin{aligned}
        f^j &= \frac{1}{\sqrt{(2\pi)^2\vert\Sigma^R\vert}}\exp\left(-\frac{1}{2}(x-\mu^j)^T(\Sigma^R)^{-1}(x-\mu^j)\right)\\ 
        &\qquad\qquad\qquad\qquad\qquad\qquad j=1,2,\ldots,N
\end{aligned}
\end{equation}
where $\vert\cdot\vert$ is the determinant and $x$ represents any point in the domain. 

The time-averaged distribution formed by the multi-agent system is then defined by the following form:
\begin{equation}\label{eqn: rho_cont}
\rho(x,t) = \frac{1}{Nt} \left(\sum_{j=1}^{N}\int f^j(x,t)dt\right),
\end{equation}
where $f^j(x,t)$ denotes a skinny Gaussian for agent $j$ in the continuous time case.

Notice that in the above form, the integral is taken with respect to time with a summation over all agents, followed by a division for the number of agents as well as the total elapsed time. Thus, it is not only the time- but also agent-averaged behavior.

The counterpart for the discrete-time case is then written by
\begin{equation}\label{eqn: rho_dis}
    \rho_{k} := \rho(x,k) =  \frac{1}{k+1}\left (\frac{1}{N}\sum_{j=1}^{N}\sum_{i=0}^{k} f_i^j\right),
\end{equation}
where $f_i^j$ stands for a skinny Gaussian in the discrete-time step for  agent $j$.
The above discrete-time equation will be used throughout the paper to derive the timing condition for ergodicity.

The difference between the time-averaged distribution $\rho_k$ and the given spatial distribution $\rho^*$ at any time $k$ is written as
\begin{align}\label{eqn: phi_{k}}
    \phi_k=\rho_k-\rho^*
\end{align}
Furthermore, the ergodic function $V_k$ is defined as the integral of the absolute value of $\phi_k$ over the given domain $\Omega$ by
\begin{align}
    V_k = \int_{\Omega}|\phi_k|dx\label{eqn: V_k}
\end{align}
The value for $V_k$ always lies between $0$ and $2$ due to the given definition of $V_k$. 


Fig. \ref{fig: problem_schematic} illustrates multiple agents with mass generation in the shape of skinny Gaussian as well as the given spatial distribution as an MoG with a negative sign as in \eqref{eqn: phi_{k}}. 
The spatial domain $\Omega$ is divided into two regions: the region consisting of the holes, denoted by $\Omega_2$ (blue dashed lines in Fig. \ref{fig: problem_schematic}) and the remaining area of $\Omega$ outside $\Omega_2$, denoted by $\Omega_1$ (red solid lines in Fig. \ref{fig: problem_schematic}). Mathematically, $\Omega_1$ and $\Omega_2$ are defined as
\begin{align}
    \Omega_1 &:= \Omega - \Omega_2, \quad
    \Omega_2 := \{x\vert x\in \bigcup_{i=1}^{m} X(\mathcal{N}\left(\mu_i, \Sigma_i)\right) \},
\end{align}
where $X(\mathcal{N}\left(\mu_i, \Sigma_i)\right)$ denotes the domain belongs to each component-wise Gaussian $\mathcal{N}\left(\mu_i, \Sigma_i\right)$.

The main objective of this paper is to achieve the ergodicity such that the time-averaged robot distribution $\rho_k$ converges to the given spatial distribution $\rho^*$ as time approaches infinity ($V_k\rightarrow 0$ as $k\rightarrow \infty$). In other words, a proper control strategy for the multi-agent system is required for the ergodicity, which is equivalent to determine the robot position $\mu_k^j$, $j=1,2,\ldots,N$, at each discrete time $k$.

One may try to achieve this goal by making the agents stay at each hole (or component-wise weighted Gaussian in a given MoG as shown in Fig. \ref{fig: problem_schematic}) with a given portion $\alpha_i$. However, the approach is too simplistic and will not work for this particular problem for the reasons that follow. Recalling the time-averaged dynamics in \eqref{eqn: rho_dis}, it can be written recursively by
\begin{align}
    \rho_{k} = \dfrac{1}{k+1}\left(k\cdot\rho_{k-1} + \frac{1}{N}\sum_{j = 1}^{N}f_{k}^j\right)\label{eqn: f_k recursive}
\end{align}
It is evident from \eqref{eqn: f_k recursive} that the influence of the current mass generation $f_k$ on the time-averaged distribution $\rho_k$ reduces nonlinearly by $\dfrac{1}{k+1}$, resulting in the difficulty to attain the ergodicity. 
Moreover, even if one hole is completely filled with mass generated by the agents, the mass vanishes gradually as soon as the agents depart from that hole. Lastly, the agents are unable to jump from one hole to another and hence, they generate unnecessary masses while traveling through the $\Omega_1$ region. These issues induce uncertainty about timing for the agents to visit each hole and the duration they should stay there. 

In the following section, we thus provide the analysis to guarantee that $V_k$ is decreasing under a certain condition.


\section{ERROR ANALYSIS OF ERGODIC OPERATION}
The proposed multi-agent exploration scheme is developed in a way that all agents act as a team to explore a hole   together and then, proceed to another hole once the current hole is filled to a certain amount. The variable $h$ is given to denote the averaged time for the agents being inside $\Omega_1$. Similarly, $h'$ indicates the averaged time spent by the agents in a hole to explore that hole. Alternatively, $h$ and $h'$ can be written by
\begin{align}
    h = \frac{1}{N}\sum_{j=1}^N h^j, \qquad h'= \frac{1}{N}\sum_{j=1}^N {h^j}'
\end{align}
where $h^j$ and ${h^j}'$ denote the time spent by agent $j$ in $\Omega_1$ and a hole, respectively. For a given spatial distribution with an MoG form having multiple holes, a subscript will be used to indicate a specific hole. Before proceeding to the error analysis, the following proposition sheds light on how the time-averaged distribution changes as the team of agents move in the domain.

\begin{proposition}\label{proposition: delta_rho_k}
Given the time-averaged distribution $\rho_k$ at any time $k$, the variation in the time-averaged distribution after $h$ time steps, $\Delta \rho_k^h$, can be calculated by
\begin{align*}
       \Delta \rho_k^h&:= \rho_{k+h}-\rho_k=\frac{1}{k+h+1}\left(\frac{1}{N}\sum_{j=1}^{N}\sum_{i=k+1}^{k+h}f_i^j-h\rho_k \right )
\end{align*}
\end{proposition}
\begin{proof}
From \eqref{eqn: f_k recursive}, the time-averaged distribution at time $k+h$ can be written as
\begin{align}\label{eqn: rho_k+h}
    \rho_{k+h} = \dfrac{1}{k+h+1}\left((k+1)\rho_{k} + \frac{1}{N}\sum_{j = 1}^{N}\sum_{i = k+1}^{k+h}f_{i}^j\right)
\end{align}


Then, the following equation can be obtained for $\Delta \rho_k^h$ using \eqref{eqn: rho_k+h}.
\begin{align*}
    \Delta\rho_k^h &=\frac{1}{k+h+1}\left(\frac{1}{N}\sum_{j=1}^{N}\sum_{i=k+1}^{k+h} f_i^j - h\rho_k \right)
\end{align*}
Similarly, $\Delta\rho_{k+h}^{h'}$ is obtained by
\begin{align*}
    \Delta\rho_{k+h}^{h'} =\frac{1}{k+h+h'+1}\left(\frac{1}{N}\sum_{j=1}^{N}\sum_{i=k+h+1}^{k+h+h'} f_i^j - h'\rho_{k+h} \right)
\end{align*}


\end{proof}

\begin{figure}
\begin{center}
\includegraphics[scale=0.45]{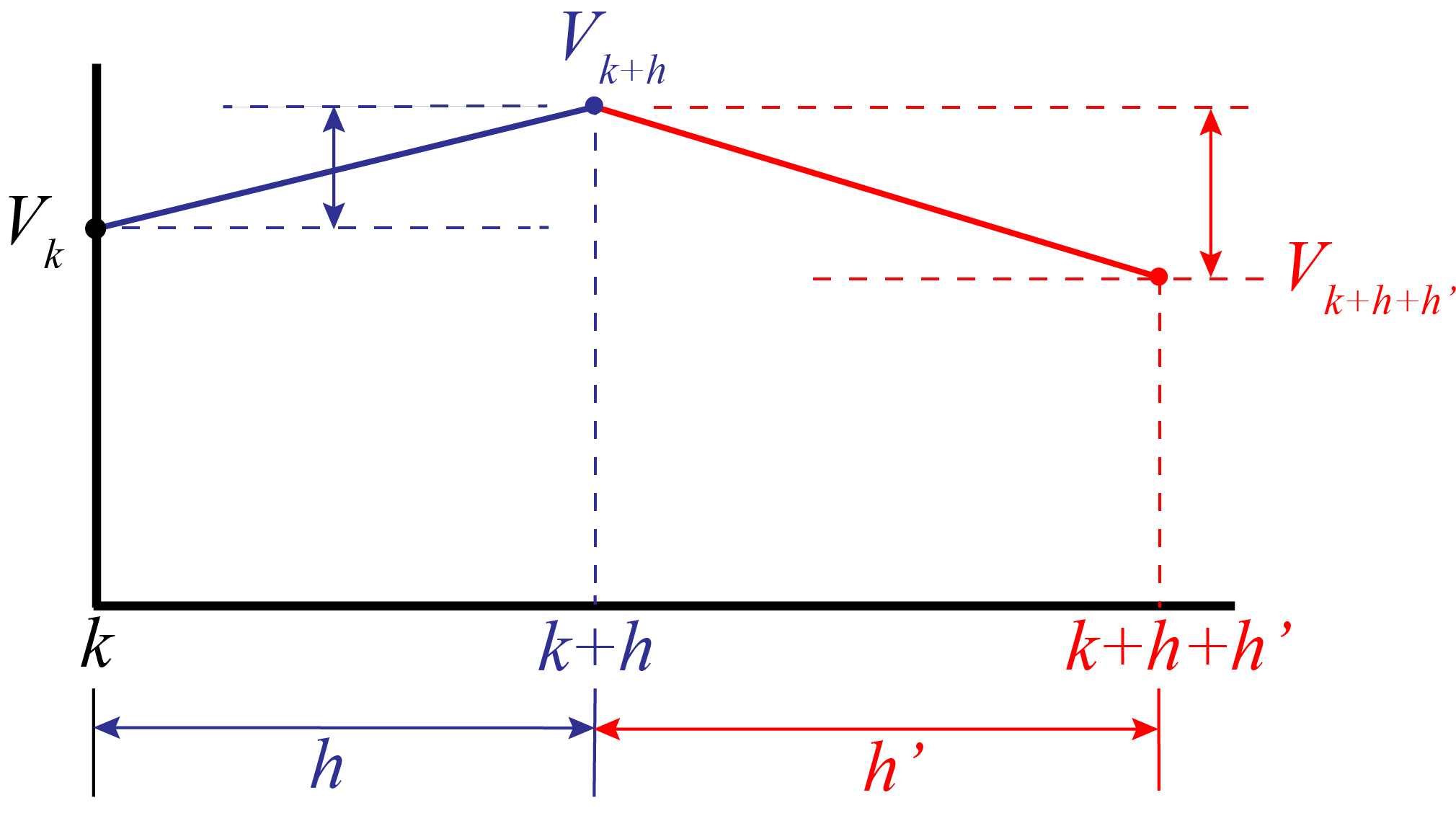}
\caption{Piece-wise variation of ergodic function $V_k$ with discrete time}\label{fig: V_plot}
\end{center}
\end{figure}

The agents need to spend time in $\Omega_1$ while traveling to a hole from the previous hole and therefore, $V_k$ goes up as the agents are spending time where they should not be. On the contrary, exploring a hole to match the time-averaged distribution with the given spatial distribution for a hole results in drop of $V_k$. Based on these observations, the proposed method to attain the ergodicity can be explained in the following way. 

While the team of agents travel from one hole to another hole with an average of $h$ time steps in $\Omega_1$, $V_k$ increases since the mass is generated outside holes. On the other hand, $V_k$ decreases while agents stay in a hole. This is illustrated in Fig. \ref{fig: V_plot}. 
Given the period for the time step $h+h'$, the piece-wise decreasing property of $V_k$ is guaranteed if the decrement of ergodic function from $V_{k+h}$ to $V_{k+h+h'}$ is greater than the increment from $V_{k}$ to $V_{k+h}$. Satisfying this condition throughout the multi-agent exploration can ensure that the ergodic function $V_k$ will converge to zero, which is defined as piece-wise convergence. To this end, the following theorem is developed for the piece-wise convergence of the ergodic function.

\begin{theorem}\label{theorem: 1}
Consider the multi-agent ergodicity problem to realize $V_k\rightarrow 0$ as $k\rightarrow \infty$. 
Given $h$ time steps for the team of agents in $\Omega_1$ to reach a certain hole, the ergodic function $V_k$ is a piece-wise contraction mapping, if the agents stay at the hole for $h'$ time steps given by
\begin{align}
     h' > \left( \frac{\int_{\Omega_2}\rho_k dx}{\int_{\Omega_1}\rho_k dx} \right )h\label{eqn: theorem_1}
\end{align}
\end{theorem}
In this case, we have $    |V_{k+h+h'}-V_{k+h}|>|V_{k+h}-V_k|
$.

\begin{proof}
Fig. \ref{fig: V_plot} illustrates that for $h$ time steps, the ergodic function $V_k$ increases such that $V_{k+h}>V_k$ and hence, $|V_{k+h}-V_k|=V_{k+h}-V_k$. Then, $|V_{k+h}-V_k|$ is obtained by
\begin{align}\label{eqn: V for h (1)}
\nonumber
|V_{k+h}-V_k|&=\int_{\Omega_1}(|\phi_{k+h}|-|\phi_k|)dx \\ &\qquad\qquad\qquad+\int_{\Omega_2}(|\phi_{k+h}|-|\phi_k|)dx 
\end{align}

For an ideal case, $\phi_k$ is positive (negative) in $\Omega_1$ ($\Omega_2$). From this observation, one can proceed with \eqref{eqn: V for h (1)} by replacing $\phi_k$ in \eqref{eqn: V for h (1)} with \eqref{eqn: phi_{k}} as
\begin{align}\label{eqn: V for h(2)}
\nonumber
   |V_{k+h}-V_k|&=\int_{\Omega_1}(\rho_{k+h}-\rho_k)dx-\int_{\Omega_2}(\rho_{k+h}-\rho_k)dx\\
   & = \int_{\Omega_1}\Delta\rho_{k}^hdx-\int_{\Omega_2}\Delta\rho_{k}^hdx
\end{align}

From Proposition \ref{proposition: delta_rho_k}, the above equation can be rewritten as
\begin{align}\label{eqn: V for h (3)}
\nonumber
  |V_{k+h}-V_k|&=\frac{1}{k+h+1}\left[ \int_{\Omega_1} \left (\frac{1}{N}\sum_{j=1}^{N}\sum_{i=k+1}^{k+h}f_i^j \right. \right. \\
    &\qquad\qquad\qquad \left. \left. -h \rho_k \right )dx + \int_{\Omega_2} h\rho_k dx  \right] 
\end{align}

In \eqref{eqn: V for h (3)}, some terms in the right hand side can be simplified as follows:
\begin{align}\label{eqn: h_omega1}
\nonumber
    &-\int_{\Omega_1}h\rho_k dx + \int_{\Omega_2}h\rho_k dx \\ &=-2h\int_{\Omega_1}\rho_k dx +  h\int_{\Omega_1+\Omega_2}\rho_k dx\\\nonumber
    &= -2h\int_{\Omega_1}\rho_k dx + h
\end{align}

\begin{align}\label{eqn: f_i omega1}
\nonumber
    \int_{\Omega_1} \frac{1}{N}\sum_{j=1}^{N}\sum_{i=k+1}^{k+h}f_i^j dx &= \frac{1}{N}\sum_{j=1}^{N}\sum_{i=k+1}^{k+h}\int_{\Omega_1}f_i^j dx \\
    & = \frac{1}{N}\sum_{j=1}^{N}\sum_{i=k+1}^{k+h}1 = h
\end{align}

Plugging \eqref{eqn: h_omega1} and \eqref{eqn: f_i omega1} into \eqref{eqn: V for h (3)}, it becomes
\begin{align}\label{eqn: V for h (4)}
   |V_{k+h}-V_k|=\frac{2h\int_{\Omega_2}\rho_k dx}{k+h+1} = \frac{2ha}{k+h+1}
\end{align}
where 
\begin{align}
\nonumber
   a:=\int_{\Omega_2}\rho_k dx = 1-\int_{\Omega_1}\rho_k dx
\end{align}

In the next step, the condition for $h'$ such that $V_{k+h+h'}<V_{k+h}$ is derived. For $h'$ time steps, the decrement of the ergodic function $V_k$ can be observed from Fig. \ref{fig: V_plot} and hence, $|V_{k+h+h'}-V_{k+h}|= -(V_{k+h+h'}-V_{k+h})$. By following the same procedure to obtain \eqref{eqn: V for h (4)}, the expression for $|V_{k+h+h'}-V_{k+h}|$ can be derived and written as
\begin{align}\label{eqn: V for h' (4)}
    |V_{k+h+h'}-V_{k+h}|=\frac{2h'\left( 1-\int_{\Omega_2}\rho_{k+h}dx\right )}{k+h+h'+1}
\end{align}
where
\begin{align}\label{eqn: 1-int rho_k_h}
    1-\int_{\Omega_2}\rho_{k+h}dx = \int_{\Omega_1}\rho_{k+h}dx
\end{align}

Note that, while the agents are travelling through $\Omega_1$, the agents do not generate any mass in $\Omega_2$ (no $f_i^j$ term). Thus, from \eqref{eqn: rho_k+h}, we have
\begin{align}\label{eqn: int rho_k+h}
    \int_{\Omega_2}\rho_{k+h} dx = \dfrac{k+1}{k+h+1}\cdot\int_{\Omega_2}\rho_{k} dx = \dfrac{(k+1)a}{k+h+1}
\end{align}
resulting in
\begin{align}\label{eqn: V for h' (5)}
    |V_{k+h+h'}-V_{k+h}| = \frac{2h'(1-\frac{(k+1)a}{k+h+1})}{k+h+h'+1}
\end{align}

To guarantee the piece-wise convergence of the ergodic function, the following condition must be satisfied:
\begin{equation}\label{eqn: V ineq}
    |V_{k+h+h'}-V_{k+h}|>|V_{k+h}-V_k|
\end{equation}

The condition guaranteeing \eqref{eqn: V ineq} can be derived by replacing both sides terms with preceding results as follows.
\begin{align}\label{eqn: V_h'>V_h}
    \nonumber
    \frac{2h'\left(1-\left(\frac{k+1}{k+h+1}\right )a \right )}{k+h+h'+1}>\frac{2ha}{k+h+1}  \,\,
    \Rightarrow  \,\, h'>\frac{a}{(1-a)}h
\end{align}

or equivalently,
\begin{equation*}
     h' > \left( \frac{\int_{\Omega_2}\rho_k dx}{\int_{\Omega_1}\rho_k dx} \right )h
\end{equation*}
\end{proof}

Theorem \ref{theorem: 1} provides the duration $h'$, the team of agents should stay in a certain hole given that the team spends $h$ amounts of time steps in $\Omega_1$. To ensure that the ergodic function $V_k$ will be piece-wise decreasing, this condition must be satisfied. 

In the proceeding section, the robot control law is provided to describe how the team of agents will explore the domain, which is different from the timing analysis presented in this section and hence, provides another contribution in this research.


\section{ROBOT CONTROL LAW}
In the previous section, the timing analysis is provided to guarantee the piece-wise convergence of $V_k$. As this result is not related to how the team of robots need to explore the holes, this section will present the robot control law that is the combination of two different control methods: the nearest point and the gradient method.

The nearest point method is to make each agent visit some point in a hole where the $\phi_k$ is negative and closest to the current agent locations. According to this control scheme, the position for the agent $j$ can be updated by
\begin{align}\label{eqn: robot dynamics}
    \mu_{k+1}^{j} = \mu_{k}^j + v_{\max}\cdot\dfrac{g^j_k-\mu_{k}^j}{\lVert g^j_k-\mu_{k}^j \rVert}
\end{align}
where $v_{max}$ is the maximum velocity attainable by the robot and $g^j_k$ denotes the nearest location with negative $\phi_k$ obtained by
\begin{align}\label{eqn: discrete g_k}
    g_k^j = \argmin_{ x\in\{x|\phi_k(x) < 0\} }||x - \mu_k^j||
\end{align}

The nearest point-based control law may lead to a successful exploration for the centralized multi-agent system; However, the method is nearsighted in that it always drives each agent towards the nearest point where $\phi_k<0$.

The second control law, the gradient method, is introduced as follows:
\begin{align}\label{eqn: gradient}
    \mu_{k+1}^{j} = \mu_{k}^j + v_{\max}\dfrac{G_k^j}{\lVert G_k^j \rVert}
\end{align}
where $G_k^j$ is the gradient of $\phi_k$ at the current location for the agent $j$. This gradient-based control law in general provides the information about magnitude and direction for the agent to fill the hole with a given mass generation. Obviously, this gradient-based method does not work efficiently once it gets stuck in either local minima or maxima.

To compensate for the weaknesses of the two different methods, we provide the combination of the two as follows:
\begin{align}\label{eqn: robot dynamics 2}
    \mu_{k+1}^{j} = \mu_{k}^j + v_{\max}\left(r\cdot \dfrac{g^j_k-\mu_{k}^j}{\lVert g^j_k-\mu_{k}^j \rVert} + (1-r)\cdot\dfrac{G_k^j}{\lVert G_k^j \rVert}\right)
\end{align}
where $r$ is defined by
\begin{align}\label{eqn: r}
    r = \frac{\rho_k(\mu_k^j)}{\rho_k(\mu_k^j) + \rho^*(\mu_k^j)}
\end{align}
Here, $\rho_k(\mu_k^j)$ and $\rho^*(\mu_k^j)$ indicate the values of time-averaged robot distribution and given spatial distribution at current agent location, i.e., $\rho_k(\mu_k^j)=\rho_k(x = \mu_k^j)$ and $\rho^*(\mu_k^j)=\rho^*(x=\mu_k^j)$. Notice that the values of $r$ always vary between $0$ and $1$ according to the given definition.

The parameters $r$ and $(1-r)$ in \eqref{eqn: robot dynamics 2} are defined as the weights assigned to the nearest point and gradient methods, which are included to help the agents decide which method to prioritize to update the position based on $\rho_k(\mu_k^j)$ and $\rho^*(\mu_k^j)$. For example, $r > (1-r)$ implies that at the current agent location $\mu_k^j$, there is excess mass ($\rho_k(\mu_k^j) > \rho^*(\mu_k^j)$, and thus, $\phi_k(\mu_k^j)>0$). So the agent is driven towards nearest point with negative $\phi_k$. On the other hand, $r < (1-r)$ indicates mass demand at $\mu_k^j$ , so the agent needs to move towards the nearest local minima of $\phi_k$ by prioritizing the gradient method. 
Thus, the agents can realize efficient exploration of the given domain by switching priorities between \eqref{eqn: robot dynamics} and \eqref{eqn: gradient} depending on the values of $r$ and $(1-r)$. 
As a consequence,  \eqref{eqn: robot dynamics 2} has higher convergence speed compared to \eqref{eqn: robot dynamics} and \eqref{eqn: gradient}.

\section{ALGORITHM}

\begin{figure}
\begin{center}
\includegraphics[scale=0.38]{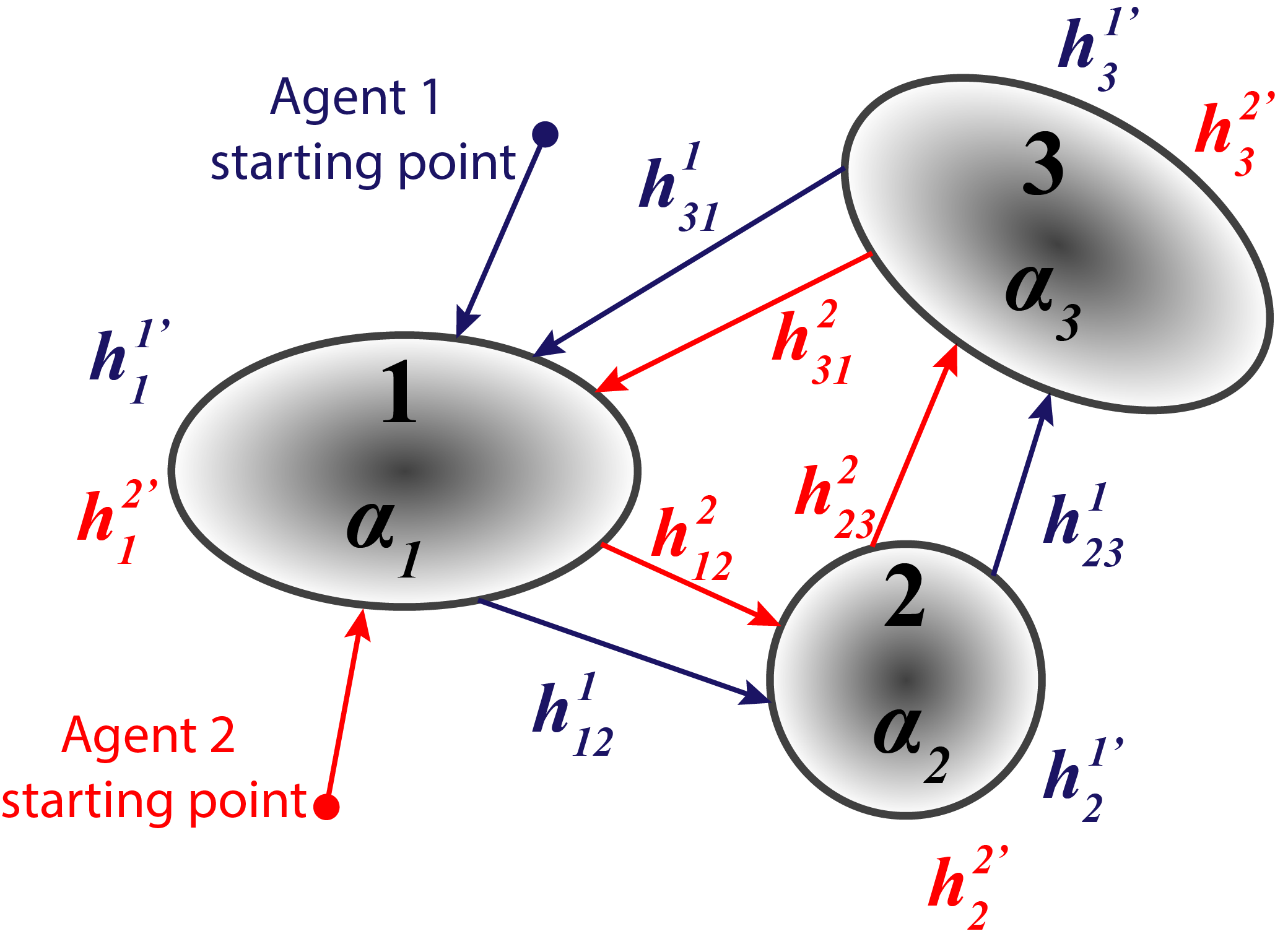}
\caption{Ergodic exploration trajectories of multi-agent system with 2 agents for 3-hole spatial distribution}\label{fig: Algorithm}
\end{center}
\end{figure}

This section provides the formal algorithm to attain the ergodicity for the multi-agent system. Fig. \ref{fig: Algorithm} illustrates how the team of agents travels in the domain $\Omega$. The red and blue points in the figure are given as the starting points for the two-agent system. Initially, the agents search for the first target hole to visit using the following steps. The value $d_{ji}$ is defined by the distance from the mean location of $i^{th}$ Gaussian component in the MoG to the $j^{th}$ agent. Then, the target hole $i^*$ for all agents to visit is determined from the following equation:
\begin{align}\label{eqn: hole detection}
    i^* = \argmin_{i} \left(\frac{\sum_{j = 1}^{N}d_{ji}}{\alpha_i}\right)
\end{align}
The above equation shows that not only distances but also the weight of each hole are taken into account to determine the first target hole. As a result, even if a hole with a greater weight is not the closest one, it may be selected to be visited initially due to its weight. 


For the given example in Fig. \ref{fig: Algorithm}, hole $1$ is found as the first target hole according to \eqref{eqn: hole detection}, as a result, the agents approach and explore the hole. The agents generate masses in the shape of skinny Gaussian distribution in every time step, as described in Assumption 2. The agents update their position using \eqref{eqn: robot dynamics 2} as discussed in the previous section.

The necessary exploration time in the current hole is proposed in \eqref{eqn: theorem_1} to guarantee the convergence of $V_k$, however, the theorem only provides the lower bound of the hole exploration time. This indicates that the convergence speed for $V_k$ may be too slow given the team of agents decides to leave the current hole just after \eqref{eqn: theorem_1} is satisfied. Consequently, another condition is required for the departure time that should be greater than the time in \eqref{eqn: theorem_1}.
The following condition is included to determine a proper departure time for the agents from the current hole:
\begin{align}\label{eqn: beta}
    h'':=\int_{\Omega_2\cap X(\text{target hole})} \vert \phi_k\vert dx > c_N
\end{align}
where $c_N = \beta \cdot e^{-\gamma\cdot N}$ with 
$\beta$ and $\gamma$ being some positive coefficients and $N$ as a cycle number. This cycle number $N$ increases when the team of agents visit all the holes and arrive at the initial hole again.

The condition in \eqref{eqn: beta} is incorporated to make the agents explore the hole until the accumulated error in the current hole $\int_{\Omega_2\cap X(\text{hole 1})} \vert \phi_k\vert dx$ obtains a value greater that $c_N$, meaning that the hole is to be filled by a mass of certain amount as mentioned in \eqref{eqn: beta}. 

The coefficient $c_N$ is defined as the above form due to the following observations.
First of all, the $\frac{1}{k+1}$ term in \eqref{eqn: rho_dis} implies that at lower time steps, the mass generated by the agents $f_k^j$ has greater contribution to the time-averaged distribution $\rho_k$ and error $\phi_k$. As a result, the longer exploration time may lead to the increment of $V_k$. Additionally, the agents cannot move outside the hole as soon as they decide to exit, which results in spilling some extra masses in the hole. Thus, the error $\phi_k$ may achieve a value that is positive, given $c_N$ is zero, which is undesired. Hence, $c_N$ is defined in a way that at lower time steps, the agents decide to leave the hole with some existing error in the hole and at later explorations, the agents exit the hole with less error in it. 

Once the exploration time for the current hole has a value higher than $\max(\Bar{h}', \Bar{h}'')$,  the team of agents heads towards the next hole, which is predetermined by the given configuration of the MoG. Here, $\Bar{h}'$ and $\Bar{h}''$ denote the time when the team of agents first satisfies \eqref{eqn: theorem_1} and \eqref{eqn: beta}. Each agent takes the shortest possible path to travel from one hole to another, although the path varies from agent to agent as the agents are located at different positions at the time of departure. In Fig. \ref{fig: Algorithm}, the agents approach hole $2$ instead of hole $3$ because hole $2$ is closer to hole $1$. 

The agents explore hole $2$ in a similar way that they explored hole $1$. They fill up the hole with mass and update their position using \eqref{eqn: robot dynamics 2}. As soon as the exploration time is greater than $\max(\Bar{h}',\Bar{h}'')$, they exit hole $2$ and travel to hole $3$. After the exploration of hole 3, the team will again approach hole $1$ and thus, they explore the domain in a cyclic manner. While traversing $\Omega_1$, the agents may not follow the same path to go from one hole to another because of the variation of $h^j_{12}$, $h^j_{23}$ and $h^j_{31}$ from different exploration cycle. Here, $h^j_{il}$ is defined as the time spent in $\Omega_1$ by agent $j$ to reach the $l^{\text{th}}$ hole from $i^{\text{th}}$ hole. 

A pseudo code is provided to illustrate the formal procedure of the proposed multi-agent centralized ergodic algorithm.

\begin{algorithm}[!h]
\caption{Multi-agent Centralized Ergodic Exploration Algorithm}\label{algorithm:1}
\begin{algorithmic}[1]
\State initialize $\rho^*$, $v_{max}$, $f_0^j$, $\mu_0^j$ $k\gets 0$
\State Find the target hole:
\If{$k=0$} it is calculated from \eqref{eqn: hole detection}
\Else{ it is updated by the given configuration of an MoG}
\EndIf
\While{the team of agents staying time in the current hole $<\max(\Bar{h}',\Bar{h}'')$}
\State \textbf{Each agent implements the following}
\For{$j \gets 1$ to $N$}
\State Calculate $r$ from \eqref{eqn: r} and $g_k^j$ from \eqref{eqn: discrete g_k}
\State Update the next agent position $\mu_{k+1}^j$ by \eqref{eqn: robot dynamics 2}
\State Fill up the target hole by generating mass $f_k^j$
\EndFor
\State Update $\rho_k$ from \eqref{eqn: rho_dis}
\State Calculate $\phi_k$ and $V_k$ from  \eqref{eqn: phi_{k}} and \eqref{eqn: V_k}, respectively
\State $k\gets k+1$
\EndWhile
\State Repeat from step 2 for the next hole

\end{algorithmic}
\end{algorithm}


\section{SIMULATIONS}
Numerical simulation results are provided in this section to verify the the correctness of the proposed methods as well as the effectiveness of the ergodic exploration algorithm. 

The spatial reference distribution is given as an MoG such that
\begin{align*}
    \rho^* &= \sum_{i=1}^{4}\alpha_i\mathcal{N}(\mu_i,\Sigma_i),\\
    \text{where }
    \mu_1 &=[120, 320]^{T}, \mu_2=[80, 100]^T, \\
    \mu_3 &= [300,120]^T, \mu_4 = [320, 320]^T \\
    \Sigma_1 &= \begin{bmatrix}
        25 & 0\\
        0 & 15
    \end{bmatrix},
    \Sigma_2 = \begin{bmatrix}
        15 & 0\\
        0 & 20
    \end{bmatrix},\\
    \Sigma_3 &= \begin{bmatrix}
        20 & 0\\
        0 & 20
    \end{bmatrix},
    \Sigma_4 = \begin{bmatrix}
        10 & 0\\
        0 & 15
    \end{bmatrix}
\end{align*}
with $\alpha = [0.3, 0.2, 0.4, 0.1]$.


The three-agent system is considered here with their initial positions given as 
$\mu_0^1 = [180,175]^T$, $\mu_0^2 = [200,300]^T$, $\mu_0^3 = [300,200]^T$. The covariance matrix for the mass generation in the form of the skinny Gaussian is $\Sigma^R=\begin{bmatrix}3 & 0\\0 & 3\end{bmatrix}$. The maximum velocity of all agents is limited to $10$.

\begin{figure*}[t]
    \centering
    \subfloat[Initial simulation setup.]{\includegraphics[scale=0.23]{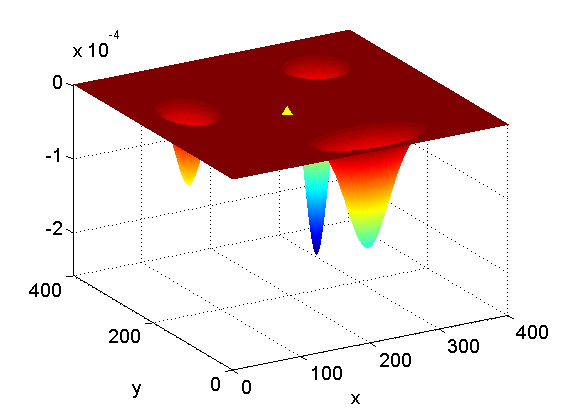}}
    \subfloat[k=20]{\includegraphics[scale=0.33]{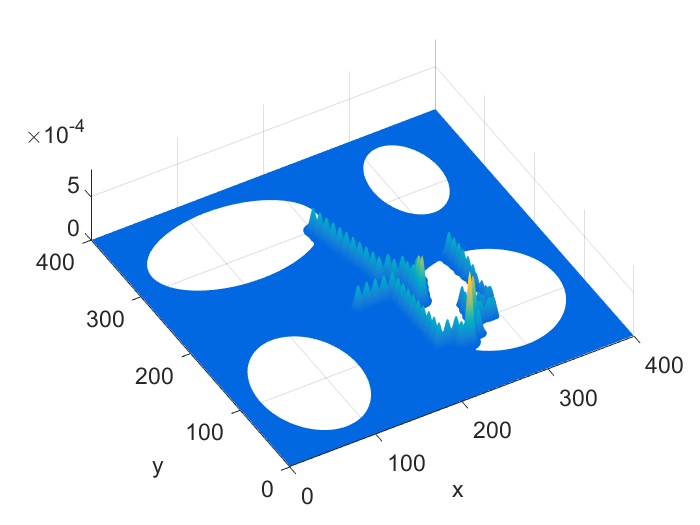}}
    \subfloat[k=2000]{\includegraphics[scale=0.4]{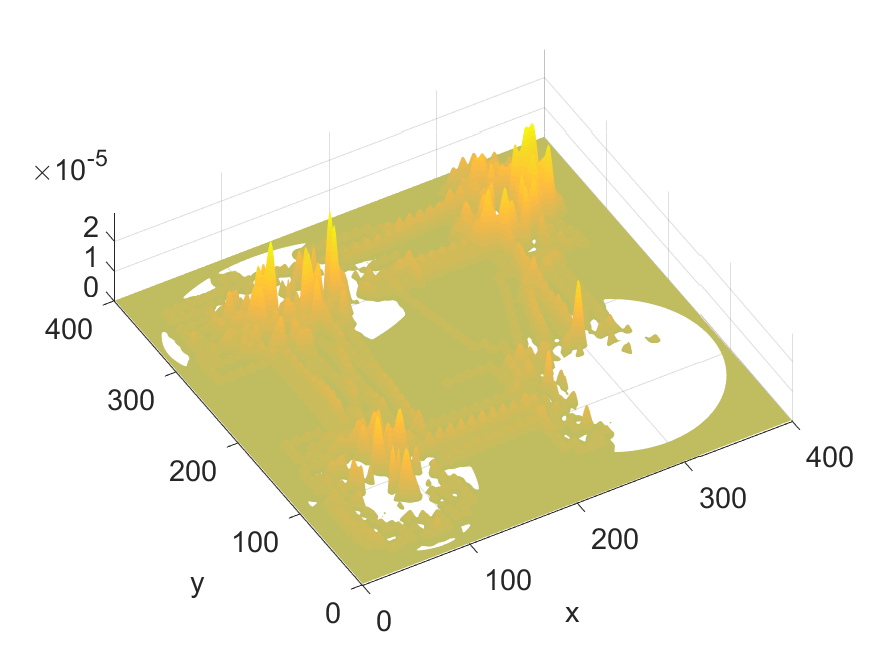}}\\
    \subfloat[k=5000]{\includegraphics[scale=0.4]{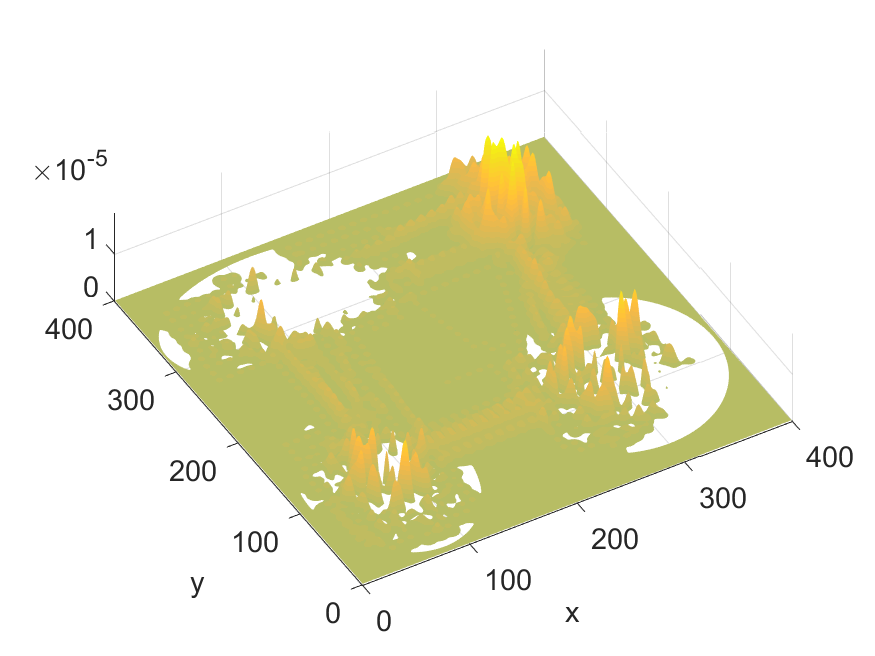}}
    \subfloat[k=7500]{\includegraphics[scale=0.4]{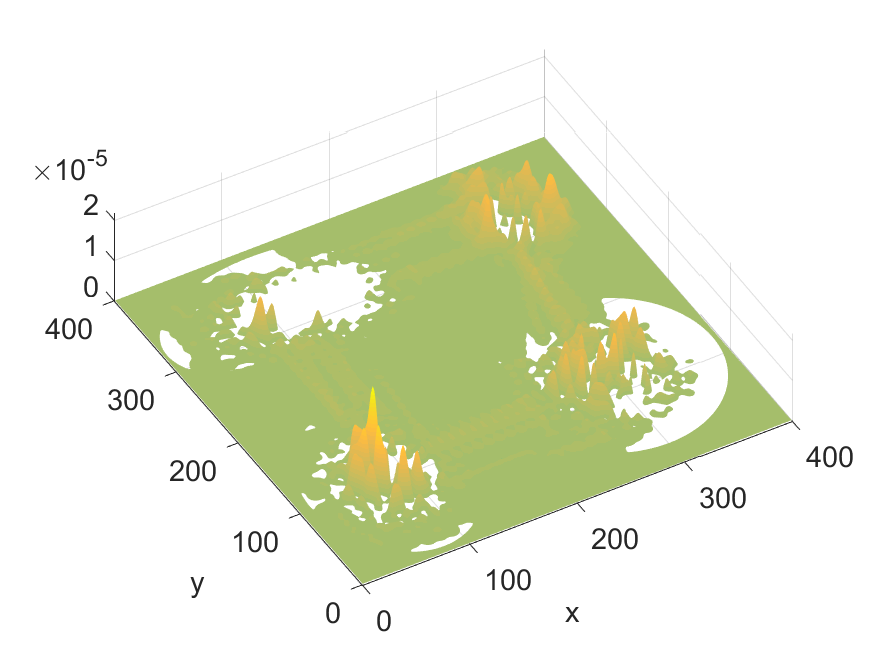}}
    \subfloat[k=10000]{\includegraphics[scale=0.4]{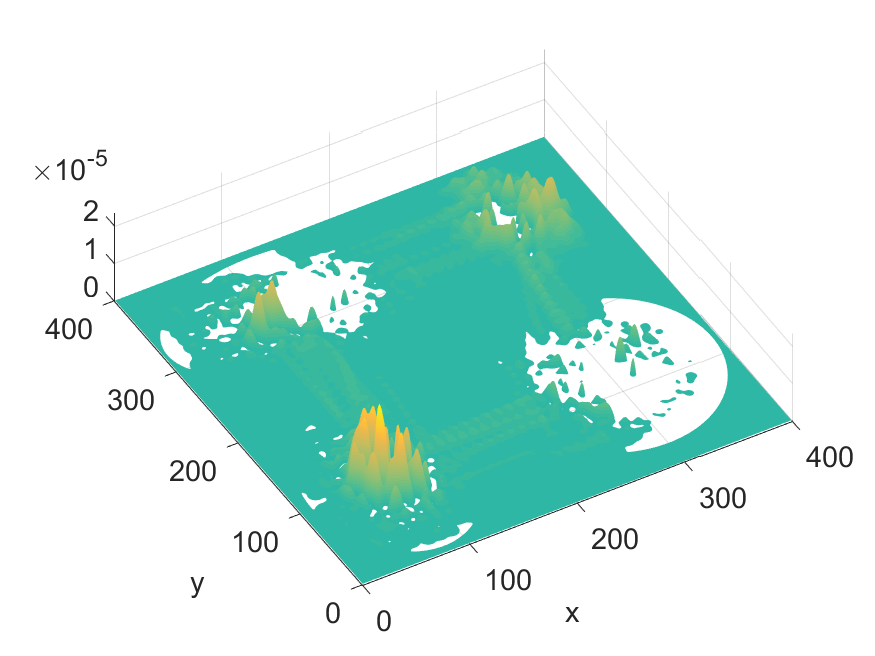}}
    \caption{The snapshot of the simulation for the ergodic exploration: (a) the spatial distribution with a negative sign with the robot initial position (red triangles) and the associated hole numbers; (b-c) snapshots of the ergodic exploration at different time steps}
    \label{fig:snapshot}
\end{figure*}

\begin{figure*}[!tbph]
    \centering
    \subfloat[]{\includegraphics[scale=0.53]{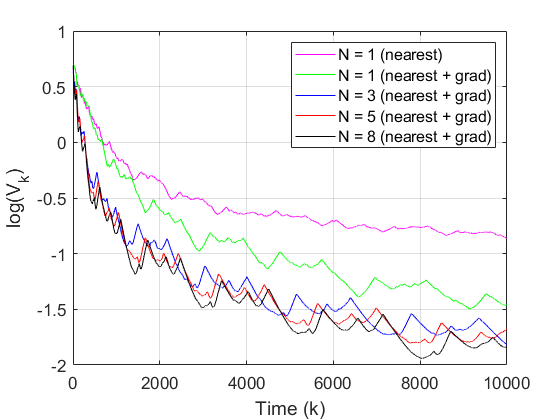}}\qquad
    \subfloat[]{\includegraphics[scale=0.53]{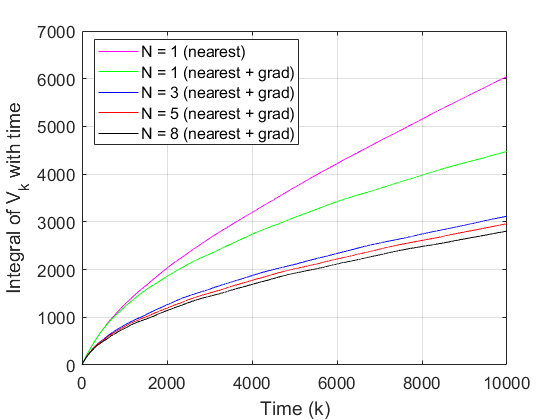}}
    \caption{The ergodic function $V_k$: (a) $V_k$ in log scale for a different number of agents as well as control method (nearest / nearest + grad); (b) integral of $V_k$ with respect to time.}
    \label{fig:Vk_log_integral}
\end{figure*}

In Fig. \ref{fig:snapshot} (a), the initial robot positions (red triangle symbols) and the negated spatial distribution along with their associated hole numbers are illustrated. Starting from the initial position, the robots head towards the initial target hole that is determined by equation (\ref{eqn: hole detection}), which in this case, is the lower-right hole (hole $3$). To satisfy the proposed convergence \eqref{eqn: theorem_1} and departure \eqref{eqn: beta} conditions, the robots spend a designated time, $h_3^{1'}$, $h_3^{2'}$, and $h_3^{3'}$, to explore the hole. After that, the robots move to the succeeding hole determined by the given configuration of the MoG, and then spend another designated time in that hole as determined by \eqref{eqn: theorem_1} and \eqref{eqn: beta}.



In Fig. \ref{fig:Vk_log_integral} (a), $V_k$ vs time plot is presented in a log scale for the convergence results. To compare the rate of convergence for $V_k$ with respect to the number of agents as well as the robot control law, numerous simulations were conducted.
For the the single agent case ($N=1$), the control law comparison is provided between the nearest point and the proposed combination method (the nearest + gradient).
It is observed that \eqref{eqn: robot dynamics 2} (labeled "nearest + grad" on figure) performs better than the nearest only control law. 
Fig. \ref{fig:Vk_log_integral} (a) does not clearly illustrate whether an increase in the agent number leads to a faster convergence of $V_k$ as the $N=$ 3, 5 and 8 cases alternate with each other. By taking the integral of $V_k$, however, it is clearly shown in Fig. \ref{fig:Vk_log_integral} (b) that an increase in the agent number yields a faster convergence of $V_k$. 
From this plot, $V_k$ keeps decreasing and hence, it can be concluded that the multi-agent system will achieve the ergodicity as $k\rightarrow\infty$.






\section{CONCLUSION}

In this paper, a centralized multi-agent exploration strategy is developed to realize ergodicity. Each agent  is assumed to generate a mass with a skinny Gaussian distribution and the reference spatial distribution is given as an MoG. The piece-wise convergence condition to attain ergodicity is derived based on the timing analysis. Multi-robot control strategy is also developed for the faster ergodic exploration and the formal algorithm for achieving ergodicity is provided. Simulations were performed to support the validity as well as effectiveness of the proposed multi-agent ergodic exploration method.










\bibliographystyle{ieeetr}
\bibliography{reference}

\end{document}